\documentclass[twoside,leqno,twocolumn]{article}
\usepackage[letterpaper]{geometry}
\usepackage[utf8]{inputenc}
\usepackage[T1]{fontenc}
\usepackage{ltexpprt}

\usepackage[nottoc, notlof, notlot]{tocbibind}
\usepackage[toc,page]{appendix}
\usepackage{blindtext}
\usepackage[colorlinks=true,citecolor=blue,linkcolor=blue]{hyperref}
\usepackage{algpseudocode, algorithm}
\usepackage{caption}

\algtext*{EndWhile}
\algtext*{EndIf}
\algtext*{EndFor}

\usepackage{tikz}

\usepackage{amsmath, amsfonts, amssymb, bbold}

\usepackage{thmtools} %load after amsthm
\newcommand{\score}[2]{Score_{#1}}
\newcommand{\rank}[2]{Rank_{#1}}

\newcommand{\cost}[2]{K(#1,#2)}

\newenvironment{algo}[1]
{\begin{figure*}[h]\rule{\textwidth}{1pt}\vspace{-.3cm}
\caption*{#1}\vspace{-.3cm}\rule{\textwidth}{.5pt}\vspace{-.5cm}}
{\vspace{-.3cm}\rule{\textwidth}{.5pt}\vspace{-.5cm}\end{figure*}}

\newcommand{\algorandomsort}
{\hyperref[algo:randomsort]{\textsc{RandomSort}}}
\newcommand{\algofootrule}
{\hyperref[algo:footrule]{\textsc{Footrule$+$}}}
\newcommand{\algoborda}
{\hyperref[algo:borda]{\textsc{Borda+}}}
\newcommand{\algosorting}
{\hyperref[algo:sorting]{\textsc{Score-then-Borda$+$}}}
\newcommand{\algoptas}
{\hyperref[algo:ptas]{\textsc{Score-then-PTAS}}}
\newcommand{\algoeptas}
{\hyperref[algo:eptas]{\textsc{Score-then-Adjust}}}

\begin{document}

\title{\textsc{How to aggregate Top-lists:
Approximation algorithms via scores and average ranks}}

\author{
  Claire Mathieu\\
  CNRS, Paris, France\\
  \texttt{clairemmathieu@gmail.com}
  \and
  Simon Mauras\\
  IRIF, Paris, France\\
  \texttt{simon.mauras@irif.fr}
}

\title{
  \Large How to aggregate Top-lists:\\
  Approximation algorithms via scores and average ranks
}
\author{
  Claire Mathieu
  \thanks{Université de Paris, IRIF, CNRS, F-75013 Paris, France}
  \\\texttt{clairemmathieu@gmail.com}
  \and
  Simon Mauras
  \footnotemark[1]
  \\\texttt{simon.mauras@irif.fr}
}

\date{}

\maketitle

%\fancyfoot[R]{\scriptsize{Copyright \textcopyright\ 2020 by SIAM\\
%Unauthorized reproduction of this article is prohibited}}
\fancyfoot[C]{\thepage}

\begin{abstract} \small\baselineskip=9pt
A top-list is a possibly incomplete ranking of elements:
only a subset of the elements are ranked, with all unranked elements tied for last.
Top-list aggregation, a generalization of the well-known rank aggregation problem,
takes as input a collection of top-lists  and aggregates them into a single complete
ranking, aiming to minimize the number of upsets (pairs ranked in opposite order
in the input and in the output).
In this paper, we give simple approximation algorithms for top-list aggregation.  
\begin{itemize}
\item\smallskip
We generalize the footrule algorithm for rank aggregation (which minimizes
Spearman's footrule distance), yielding a simple 2-approximation algorithm
for top-list aggregation. 
\item\smallskip
Ailon's \textsc{RepeatChoice} algorithm for bucket-orders aggregation yields a
2-approximation algorithm for top-list aggregation. Using inspiration from
approval voting, we define the score of an element as the frequency with which
it is ranked, i.e. appears in an input top-list. We reinterpret \textsc{RepeatChoice}
for top-list aggregation as a randomized algorithm using variables whose
expectations correspond to score and to the average rank of an element given that it is ranked.
\item\smallskip
Using average ranks, we generalize and analyze Borda's algorithm for rank
aggregation. We observe that the natural generalization is not a constant approximation.
\item\smallskip
We design a simple 2-phase variant of the Generalized Borda's algorithm,
roughly sorting by scores and breaking ties by average ranks,
yielding another simple constant-approximation algorithm for top-list aggregation.
\item\smallskip
We then design another 2-phase variant in which in order to break ties we use,
as a black box, the Mathieu-Schudy PTAS for rank aggregation,
yielding a PTAS for top-list aggregation. This solves an open problem posed by Ailon.
\item\smallskip
Finally, in the special case in which all input lists have length at most $k$,
we design another simple 2-phase algorithm based on sorting by scores,
and prove that it is an EPTAS --
the complexity is $\mathcal O(n \log n)$ when $k=o(\log n)$.
\end{itemize}
\end{abstract}

%\setcounter{tocdepth}{1}
%\tableofcontents

%\thispagestyle{empty}
%\clearpage
%\setcounter{page}{1}

\section{Introduction}
\subsection{Context.}
Rank aggregation is a classical problem in combinatorial optimization, where
the goal is to take elements from a ground set (\emph{candidates}) and find a ranking 
which is ``closest'' to a set of input rankings (\emph{voting profile}).
Rank aggregation comes up in machine learning \cite{cohen1998learning},
natural language processing \cite{li2014learning}, bio-informatics \cite{li2017comparative}, 
and is relevant in the field of information retrieval
(meta search and spam reduction \cite{dwork2001rank}, similarity search \cite{fagin2003efficient} and more).
Historically, rank aggregation was first
studied in social choice theory, where the underlying properties of
a ranking method are of critical importance \cite{borda1781memoire,condorcet1785essai,arrow1951social}.
In this paper, we use terminology  (candidates, votes, voting profile, ...) derived from social choice theory.

\subsection{Rank aggregation.}
There are several ways to measure how close the output ranking is to the input rankings:
the most popular is Kendall's tau distance, that has several satisfying structural
properties~\cite{kemeny1959mathematics,kemeny1962mathematical,young1978consistent}.
In this paper, we focus on Kendall's tau distance, that counts the number of pairs
of candidates that are ranked in reverse order in the two rankings.
Rank aggregation is NP-hard~\cite{bartholdi1989voting,dwork2001rank},
but one of the simplest randomized algorithm yields a constant factor approximation:
algorithm \textsc{Random} simply takes a random input ranking and declares it to be the output.
Many other constant factor approximation algorithms are known:
the Footrule algorithm \cite{dwork2001rank};
the randomized \textsc{KwikSort} algorithm and variants \cite{ailon2005aggregating, ailon2008aggregating},
derandomized in \cite{van2009deterministic};
Borda's method  \cite{coppersmith2010ordering};
Copeland's method, the median rank algorithm and more \cite{fagin2016algorithmic}.
There is even a polynomial-time approximation scheme~\cite{kenyon2007rank,mathieu2009rank},
but work on constant factor approximations nevertheless continued:
they are popular because of their simplicity.
Experimental studies can be found e.g. in \cite{coleman2009ranking}
(algorithms inspired by standard sorting algorithms, and local search algorithms);
\cite{schalekamp2009rank} (Footrule, Markov chain algorithms,
sorting algorithms, local search algorithms, hybrid algorithms, and more);
\cite{ali2012experiments} (additionally includes exact LP-based and branch-and-bound
algorithms as well as various heuristics, with a focus on social choice theory). 

\subsection{From full-ranking to top-lists.}
A meta search engine aggregates information from different search engines to
answer users' requests with a ranked selection of web-pages. In such settings,
a useful extension of the rank aggregation problem is to deal with incomplete
data, where each vote provides, not a full-ranking of all candidates, but an
ordered selection of his preferred candidates, the remaining ones being
implicitly tied at the end.
%``{\em any method for rank aggregation for Web applications must be capable
%of dealing with the fact that only the top few hundred entries of each ranking
%are available.}`
Such a partial ranking is called a top-list. 
Incomplete rankings were studied in \cite{dwork2001rank}, but for the most part
without the assumption that candidates that do not appear are implicitly ranked
after the candidates that appear in the list, so the input model is different. Top-list aggregation also comes up in bio-informatics: ``{\em In previous research, attention has been focused on aggregating full lists. However, partial and/or top ranked lists are prevalent because of the great heterogeneity of genomic studies and limited resources for follow-up investigation.}" \cite{li2017comparative}. 
There is an extensive discussion about distances between top-lists
in~\cite{fagin2003conference,fagin2003journal} (and more generaly
between bucket-orders in~\cite{fagin2004conference,fagin2006journal}),
and they propose several aggregation problems.\footnote{As a side remark, we observe that if we are only interested in the first
few candidates of the output full-ranking, computing a top-list which is closest
to the input top-lists might not be a good idea: there are instances where this has undesirable artefacts.} Experimental studies (for a different, related objective) can be found for example in~\cite{cohen1998learning} for top-30-lists. 

\subsection{Top-list aggregation.}
In this paper we study the top-list aggregation problem (\textsc{Top-Agg})
that takes top-lists as input. The goal is to find an full-ranking that minimizes
the average distance to a top-list from the input. We use a natural generalization of
Kendall's tau distance: we still count the number of pairs of candidates that
are ranked in reverse order in the two rankings, without counting pairs of
candidates that are tied in one ranking.
  This problem was defined by Ailon in \cite{ailon2007aggregation,ailon2010aggregation},
where he showed that \textsc{Top-Agg} is NP-Hard even if each top-list rank exactly
two candidates.
% and where he consider the further generalization of bucket-order
%aggregation (where each input ranks group of tied candidates).
Some approximation algorithm for full-ranking aggregation extend to top-list
aggregation\footnote{They actually extend to bucket orders, a further generalization, see end of section.}: Algorithm \textsc{RepeatChoice} (Ailon's generalization of
algorithm \textsc{Random}) is a 2-approximation.
Algorithm \textsc{KwikSort}
(introduced in \cite{ailon2005aggregating,ailon2008aggregating},
improved in \cite{ailon2007aggregation,ailon2010aggregation},
and determinized in \cite{van2009deterministic}) also extends to top-list aggregation,
and one of its variants yields a 3/2 approximation algorithm: candidates are ranked using a quick-sort like approach and a
randomized rounding of the relaxation of an integer-LP. 

\subsection{Our results.}
We study whether other approximation algorithms for rank aggregation can be
extended to aggregate top-lists.
The Footrule Algorithm~\cite{dwork2001rank} is an intuitive way to aggregate
full-rankings into a single full-ranking, using Spearman's footrule distance
(which approximates Kendall's tau distance, and is much easier to minimize).
We use a natural generalization of the footrule distance: the generalized
footrule distance between two partial orders is the minimum footrule distance
between their linear extensions (similar distance are discussed in
\cite{fagin2006journal}). This enables us to extend the result from~\cite{dwork2001rank}.

\begin{restatable*}{theorem}{thmfootrule}
  \label{theorem:footrule}
  Algorithm~\algofootrule\ is a 2-approximation for \textsc{Top-Agg}.
  Its running time is linear in the size of the input,
  and cubic in the number of candidates.
\end{restatable*}

Since Ailon's work~\cite{ailon2010aggregation}, a simple 2-approximation for
\textsc{Top-Agg} was already known:
algorithm \textsc{RepeatChoice} is a randomized algorithm
that was designed in the more general setting of bucket-orders.
A close-up look at \textsc{RepeatChoice} in the context of \textsc{Top-Agg}
reveals that it can be reinterpreted to use random variables whose averages are
related to the scores and average ranks of the candidates.
The score of a candidate is the frequency with which he is ranked in the input,
and his average rank is the average value of his rank when he is ranked. 

\begin{restatable*}{theorem}{thmrandomsort}
  \label{theorem:randomsort}
  Algorithm \algorandomsort\ (the specialization of Algorithm
  \textsc{RepeatChoice} from \cite{ailon2010aggregation} to \textsc{Top-Agg})
  is a randomized 2-approximation algorithm.
  Its running time is quasi-linear in the size of the input.
\end{restatable*}

In the context of full-ranking, sorting candidates by average rank is precisely
Borda's voting method~\cite{borda1781memoire}, another simple and
popular algorithm, known to be a 5-approximation~\cite{coppersmith2010ordering}.
This leads us to analyze the generalization of Borda's algorithm to \textsc{Top-Agg},
where a preprocessing step eliminates all candidates with zero scores such that
the average rank of each candidate is well defined.

\begin{restatable*}{theorem}{thmborda}
  \label{theorem:borda}
  Algorithm~\algoborda\ is a $(4\alpha+2)$-approximation algorithm for \textsc{Top-Agg},
  where $\alpha$ is the ratio between the maximum and minimum scores of candidates
  (assuming that all scores are non-zero).
  Its running time is quasi-linear in the size of the input.
\end{restatable*}

Unfortunately, Algorithm~\algoborda\ is not an $\mathcal O(1)$-approximation for
\textsc{Top-Agg} (see Section \ref{section:borda} for a counterexample).
This indicates that the score of a candidate is of primary importance.
We therefore design a slightly less simple, but still elementary 2-phase algorithm,
that first roughly sorts by scores, putting into buckets candidates that have
similar scores, then refines the ordering using average rank, yielding
a new constant factor approximation. 

\begin{restatable*}{theorem}{thmsorting}
  \label{theorem:sorting}
  Algorithm \algosorting\ is a randomized $(8e+4)$-approximation algorithm,
  which only uses the scores and average ranks of the candidates.
  Its running time quasi-linear in the number of candidates.
\end{restatable*}

The proof of Theorem \ref{theorem:sorting} relies on a critical lemma
(Lemma \ref{lemma:partition}) proving that there exists a near-optimal
full-ranking that respects the rough ordering by scores.
As for aggregation ranking, not only are there $\mathcal O(1)$-approximations,
but there also exists a polynomial time approximation scheme~\cite{mathieu2009rank}.
Building on the intuition acquired so far, it is now easy to generalize that result,
designing a 2-phase algorithm, that first roughly sorts by scores,
then refines the ordering using the full-ranking aggregation PTAS, yielding a PTAS for
\textsc{Top-Agg}.

\begin{restatable*}[PTAS for \textsc{Top-Agg}]{theorem}{theoremptas}
\label{theorem:ptas}
For all fixed $\varepsilon > 0$, Algorithm \algoptas\ is a randomized
$(1+\varepsilon)$-approximation algorithm for \textsc{Top-Agg}. Its time complexity is
$\mathcal O\left(\frac{1}{\varepsilon} \cdot n^3 \log n \right)
+ n \exp(\exp(\mathcal O(\frac{1}{\varepsilon})))$,
the algorithm can be derandomized with an additional cost of
$\exp(\log n \exp(\mathcal O(\frac{1}{\varepsilon}))$.
\end{restatable*}

To summarize, several criteria come into play when designing an algorithm for
\textsc{Top-Agg}: approximation, running time, and simplicity, and there is a
trade-off between those. Our contribution is to explore the spectrum of existing
approximation algorithms for rank aggregation, and generalize them to map out
possible approximation algorithms for \textsc{Top-Agg}. 

\medskip
Now, remembering our initial motivations (applications to information retrieval),
one might notice that in some practical cases the total number of candidates is several
order of magnitude above the number of candidates ranked in each input top-list.
So far, we have focused on simplicity and on quality of approximation.
For the \textsc{Top-$k$-Agg} problem (when each input top-lists ranks $k$ candidates),
we can actually get both at the same time:
we design a very simple algorithm that is an efficient PTAS.
Intuitively, when $k$ is constant, for candidates that will be
ranked quite far, the average rank matters little and the score
is most important. Therefore sorting by score produces a near-optimal ranking,
except for the first few candidates, hence Algorithm~\algoeptas.

\begin{restatable*}[EPTAS for \textsc{Top-$k$-Agg}]{theorem}{theoremeptas}
\label{theorem:eptas}
For all fixed $\varepsilon > 0$, Algorithm \algoeptas\ is a $(1+\varepsilon)$-approximation algorithm
for \textsc{Top-$k$-Agg}. Its time complexity is $\mathcal O(n \log n + m \cdot 2^m$,
with $m := \lceil (1+\frac{1}{\varepsilon})(k-1)\rceil$.
\end{restatable*}

Thus, in addition to a variety of simple approximation algorithms, we provide two
approximation scheme, both solving an open problem stated in \cite{ailon2010aggregation}.

\subsection{Bucket orders.}
A further generalization of rank aggregation is obtained by letting the input
consist of bucket-orders (where a bucket-order is an ordered
partition of candidates into equivalence classes). This was considered by
\cite{fagin2006journal} and studied by Ailon \cite{ailon2010aggregation}
who gave two approximation algorithms: \textsc{RepeatChoice} and \textsc{KwikSort}.
The scenery of potential constant factor approximations for the
problem still remains to be done, and we leave the existence of an approximation
scheme for that generalization as an outstanding open problem.

\section{Definitions}
\label{section:defs}
Let $[n] := \{1, \dots, n\}$ be the set of candidates.
%\begin{notation}[Full-Ranking]
%  Let $\mathfrak S_n$ be the set of bijection from $[n]$ to itself. 
%\end{notation}

\begin{definition}[Full-ranking, Top-$k$-List, Top-list]
Let $k \in [n]$. A \emph{top-$k$-list} $\pi$, to each candidate $i\in [n]$, assigns a rank $\pi_i \in [k]\cup \{\infty\}$ such that  there is exactly one candidate of each rank $1,2,\ldots ,k$. A \emph{top-list} is a top-$k$-list for some $k$.

\medskip\noindent 
The set of top-$k$-lists is denoted $\mathfrak T_n^k$ and the set of top-lists
is denoted $\mathfrak T_n$. 

\medskip\noindent For $k=n$, a top-$n$-list is also called a \emph{full-ranking} and the set is
denoted $\mathfrak S_n$ (also called the set of permutations over $[n]$).

\medskip\noindent
A \emph{top candidate} is a candidate $i$ such that $\pi_i<\infty$.

\end{definition}

%For example, if we have ten candidates (competitors) numbered $1$ through $10$, and the gold, silver and bronze medals are awarded to candidates 2,6,3 respectively, then $k=3$ and the corresponding top-$3$-list is $\pi(2)=1,\pi(6)=2,\pi(3)=3$ and $\pi(i)=\infty$ for all other values of $i$.
For example, if we have $n=8$ candidates and $k=3$ ranks with the gold, silver and bronze medals given to candidates $2,5$ and $1$ respectively, the corresponding top-3-list is written $\pi = [2,5,1; \dots]$.
This top-list can be represented as in Figure~\ref{figure:top3list}, with candidates listed by order of rank, and candidates with rank $\infty$ listed in arbitrary order.
%\item In sections ? to ? we only need that for all $1 \leq k \leq n$ and $\pi \in \mathfrak T_n^k$ we have $|\pi^{-1}(\top)| \leq k$.
%\item In section ? we need that for all $\pi \in \mathfrak T_n$ the restriction $\pi_{|\pi^{-1}(\top)}$ is a bijection.
  % so those last four candidates could be drawn in any order and there are therefore several representations associated to the same top-$k$-list.
%A \emph{top-list} is a top-$k$-list for some $k \in [n]$.

\begin{figure}[h!]
%Having $\pi(1) = 3$ means candidate $1$ is ranked third.
  \begin{center}
    \begin{tikzpicture}[scale=.75]
      \draw (10,0) node {$\pi \in \mathfrak T_8^3$};
      \draw (0,-.1) -- (9,-.1);
      \draw (0,.1) -- (9,.1);
      \draw (0.0,-.1) -- (0.0, .1);
      \draw (1.5,-.1) -- (1.5, .1);
      \draw (3.0,-.1) -- (3.0, .1);
      \draw (4.5,-.1) -- (4.5, .1);
      %\draw (6.0,-.1) -- (6.0, .1);
      %\draw (9,-.1)  -- (9, .1);
%      \draw (0.75,.4) node {$\overbrace{\hspace{1.4cm}}^{\pi^{-1}(\{1\})}$};
      \draw (0.75,.4) node {\footnotesize$\pi_2 = 1$};
%      \draw (2.25,.4) node {$\overbrace{\hspace{1.4cm}}^{\pi^{-1}(\{2\})}$};
%      \draw (3.75,.4) node {$\overbrace{\hspace{1.4cm}}^{\pi^{-1}(\{3\})}$};
%      \draw (5.25,.4) node {$\overbrace{\hspace{1.4cm}}^{\pi^{-1}(\{4\})}$};
      \draw (2.25,.4) node {\footnotesize$\pi_5 = 2$};
      \draw (3.75,.4) node {\footnotesize$\pi_1 = 3$};
      \draw (6.75,.4) node {\footnotesize$\pi_3,\pi_4,\pi_6,\pi_7,\pi_8=\infty$};
      \fill (0.75,0) circle (1.5pt);
      \fill (2.25,0) circle (1.5pt);
      \fill (3.75,0) circle (1.5pt);
      \fill (5.25,0) circle (1.5pt);
      \fill (6.0,0) circle (1.5pt);
      \fill (6.75,0) circle (1.5pt);
      \fill (7.5,0) circle (1.5pt);
      \fill (8.25,0) circle (1.5pt);
      \draw (0.75,-.5) node {\small$2$};
      \draw (2.25,-.5) node {\small$5$};
      \draw (3.75,-.5) node {\small$1$};
      \draw (5.25,-.5) node {\small$3$};
      \draw (6.0,-.5) node {\small$4$};
      \draw (6.75,-.5) node {\small$6$};
      \draw (7.5,-.5) node {\small$7$};
      \draw (8.25,-.5) node {\small$8$};
    \end{tikzpicture}
    \caption{Representation of a top-3-list $\pi = [2,5,1; \dots]$}
    \label{figure:top3list}
  \end{center}
\end{figure}
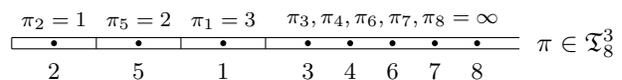

\begin{definition}[Kendall's tau distance]
\label{definition:kendalltau}
The \emph{generalized Kendall's tau distance} $K(\sigma, \pi)$ between a full-ranking $\sigma$ and a top-list $\pi$ is the number of pairs of candidates that are ranked in reverse order in $\sigma$ and in $\pi$, i.e.
\begin{equation*}
K(\sigma, \pi) := \sum_{i \in [n]}\sum_{j \in [n]} \mathbb 1_{\sigma_i > \sigma_j} \cdot \mathbb 1_{\pi_i < \pi_j}
\end{equation*}
where $\mathbb 1_P$ denotes the indicator function. When no candidate has rank $\infty$ (i.e. when $\pi$ is a full-ranking),
this definition coincides with Kendall's tau distance between two full-rankings.
\end{definition}

A pair $\{i,j\}$ of candidates that are tied in $\pi$ does not contribute to $K(\sigma, \pi)$.
Thus, considering the full-ranking $\tau$ which is a linear extension of $\pi$ where ties are broken according to $\sigma$, the generalized Kendall's tau distance between $\sigma$ and $\pi$ is exactly Kendall's tau distance between $\sigma$ and $\tau$. (We note that this is different from the distances discussed in~\cite{fagin2003conference, fagin2003journal, fagin2004conference}, where breaking ties incurs a non-zero cost.)

\begin{figure}[h!]
  %\begin{center}
    \begin{tikzpicture}[scale=.75]
      \draw (8.5, 0) node[anchor=west] {$\sigma \in \mathfrak T_8^8 = \mathfrak S_8$};
      \draw (8.5, 1) node[anchor=west] {$\pi \in \mathfrak T_8^4$};
      \foreach \y in {-.1,.1,.9,1.1} {\draw (0,\y) -- (8.2,\y);}
      \foreach \x in {0,1,...,8}{\draw (\x,-.1) -- (\x,.1);}
      \foreach \x in {0,1,...,4}{\draw (\x,0.9) -- (\x,1.1);}

      \foreach \x in {.5,1.5,..., 7.5}{\fill (\x,0) circle (1.5pt);}
      \foreach \x in {.5,1.5,..., 7.5}{\fill (\x,1) circle (1.5pt);}
      \foreach \a/\b in {4.5/1.5,7.5/7.5,2.5/0.5,0.5/2.5,1.5/4.5, 
                         3.5/6.5,5.5/3.5,6.5/5.5}{
        \draw (\a,1) -- (\b,0);
      }
      \draw (3/2,1/2) circle (2pt);
      \draw (1.5+3/5,4/5) circle (2pt);
      \draw (1.5+3/5,1/5) circle (2pt);
      \draw (3,1/2) circle (2pt);
      \draw (4,5/6) circle (2pt);
      \draw (3.9,1/5) circle (2pt);
      \draw (4.7,0.6) circle (2pt);
      \draw (5.75,1/4) circle (2pt);
    \end{tikzpicture}
  \caption{Representation of the generalized Kendall's tau distance between $\pi$ and $\sigma$. Here $K(\sigma, \pi) = 8$, and the eight pairs that contribute to the cost are materialized by the eight circles.}
  \label{figure:kemeny}
  %\end{center}
\end{figure}
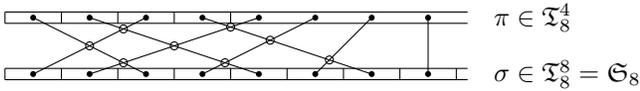

Using the graphical representation of a top-list where candidates with rank $\infty$ are listed using their order in $\sigma$, we can represent $K(\sigma, \pi)$ as in Figure~\ref{figure:kemeny}. Each candidate $i$ is associated to a line segment connecting the position of $i$ in the representation of $\pi$ and of $\sigma$, and each crossing pair $\{ i,j\}$  that contributes towards $K(\sigma, \pi)$ is marked by a small circle at the intersection of the two corresponding line segments.

%Note that for any full-ranking $\sigma$ and top-$k$-list $\pi$, we have  $0 \leq K(\sigma, \pi) \leq k (k-1) / 2 + k (n-k) \leq kn $.

\begin{definition}[Voting profile]
  A \emph{voting profile} is a distribution $p$ over top-lists. 
  The distance between a full-ranking $\sigma$ and a voting profile $p$ is the average distance
  between $\sigma$ and a top-list sampled from $p$.
  \begin{align*}
  K(\sigma, p) :=& \sum_{\pi \in \mathfrak T_n} p(\pi) \cdot K(\sigma, \pi)\\
   =& \sum_{i \in [n]}\sum_{j \in [n]} \mathbb 1_{\sigma_i > \sigma_j}
  \cdot p(\pi_i < \pi_j)\nonumber
  \end{align*}
  We denote by $p(E)=\sum_{\pi \in E} p(\pi)$ the probability of an event $E \subseteq \mathfrak T_n$. We also use the notation
  $p(\mathrm{Property~on~}\pi) := p(\{\pi \in \mathfrak T_n \;|\; \mathrm{Property~on~}\pi\})$.
\end{definition}

Equivalently, the reader may consider that a voting profile is a set of top-lists with weights.
The size of a voting profile is the sum of sizes of the top-lists in its support.

\begin{definition}[\textsc{Top-Agg} problem]
The top-list aggregation problem \textsc{Top-Agg} takes as input a set of candidates $[n]$
and a voting profile $p$, and outputs a full-ranking $\sigma$ of the $n$ candidates.
The goal is to minimize the distance $K(\sigma, p)$: the weighted average value of the generalized Kendall's tau distance
between $\sigma$ and a top-list $\pi$ from $p$.
%When there exists a $k$ such that every input top-list is a top-$k$-list, the problem is called  \textsc{Top-$k$-Agg}.
\end{definition}

%The problem \textsc{Top-2-Agg} is already NP-Hard. Therefore our goal is to
%establish approximation algorithms, with a tradeoff between simplicity and
%the approximation ratio. 

\begin{figure}[h]
Let $p$ be a voting profile such that:
\vspace{-.2cm}

\[\pi_1 = [3,5,1,7;\dots] \qquad p(\pi_1) = 1/10\]
\[\pi_2 = [3,1,4,5;\dots] \qquad p(\pi_2) = 2/10\]
\[\pi_3 = [4,1,5,2;\dots] \qquad p(\pi_3) = 3/10\]
\[\pi_4 = [6,1,2,3;\dots] \qquad p(\pi_4) = 4/10\]
The optimal solution is $\sigma^* = [1,2,3,4,5,6,7,8]$:
\vspace{-.2cm}
\[K(\sigma^*,p) = \frac{1}{10}\cdot 8 + \frac{2}{10} \cdot 4
+ \frac{3}{10} \cdot 5 + \frac{4}{10} \cdot 5 = 5.1 \]
%\vspace{-.6cm}
\caption{Example of instance of \textsc{Top-$k$-Agg} with $k=4$ and $n=8$.}
\label{figure:instance}
\end{figure}
In Figure~\ref{figure:instance} we give an instance of \textsc{Top-Agg}
that will be reused in the next sections. Observe that $K(\sigma,\pi_1) = 8$
is represented in Figure~\ref{figure:kemeny}. The optimal solution ranks candidate
$1$ first, since he is preferred to every other candidate
(this property is known as Condorcet's criterion).

In the input top-lists of Figure~\ref{figure:instance}, observe that candidate $8$
is never a top candidate, so the optimal solution ranks it last. In the upcoming algorithms, we are often going to assume without loss of generality 
that no such candidates exist, since they may be eliminated in a preprocessing step.

\section{Generalized footrule algorithm}
\label{section:footrule}
\begin{algo}
{\textbf{Algorithm} \textsc{Footrule$+$} (Generalization of \textsc{Footrule})}
  \label{algo:footrule}
  \begin{algorithmic}
    \State {\bf Input: } instance $(n,p)$ of \textsc{Top-Agg}
    \State For each candidate $i \in [n]$ and each rank $j \in [n]$:
    \State\quad  Define the cost of putting $i$ at rank $j$ as 
    $C(i,j) := \sum_{r=1}^j (j - r) \cdot p(\pi_i = r)$.
    \State Use min-cost-perfect-matching to assign candidate $i$ to rank $\sigma_i$,
    minimizing $\sum_{i=1}^n C(i,\sigma_i)$.
    \State {\bf Output} the resulting full-ranking $\sigma$.
  \end{algorithmic}
\end{algo}

Spearman's footrule distance between two full-rankings $\sigma$ and $\tau$
is the sum of displacement of each candidate:
$F(\sigma, \tau) = \sum_{i=1}^n |\sigma_i - \tau_i|$.
Diaconis and Graham showed in \cite{diaconis1977spearman} that distances $K$ and $F$
are always within a constant factor of each other:
$K(\sigma,\tau) \leq F(\sigma, \tau) \leq 2 K(\sigma, \tau)$.
Thus, approximating with respect to one distance also yields an approximation with respect to the other distance. 

Dwor, Kumar, Naor and Sivakumar noticed this fact in \cite{dwork2001rank},
and proved that minimizing $F$ can be done in polynomial; which yields a
2-approximation for full-ranking aggregation with $K$.
The algorithm computes the cost induced by ranking candidate $i$ at rank
$j$, then uses a minimum-cost-perfect-matching algorithm
to assign candidates to ranks.
Algorithm~\algofootrule\ is a generalization of the approach.

Using Algorithm~\algofootrule\ on the instance from Figure~\ref{figure:instance},
we obtain a full-ranking $\sigma = [4,1,2,3,5,6,7,8]$
which is at a distance $K(\sigma, p) = 5.8$ from $p$.
Observe that candidate $1$ is ranked second
instead of first, which would have been optimal with respect to $K$.

\thmfootrule %restatable, defined in intro.tex

\begin{proof}
Let $(n,p)$ be an instance of \textsc{Top-Agg},
and let $\sigma$ the output of Algorithm~\algofootrule.

To define a generalized version of Spearman's footrule between a
full-ranking $\sigma$ and a top-list $\pi$, we use the linear extension
$\tau$ of $\pi$ in which ties are broken according to $\sigma$:
$F(\sigma, \pi) := F(\sigma, \tau) =  \sum_{i=1}^n |\sigma_i - \tau_i|$.
As noticed in section~\ref{section:defs}, we also have 
$K(\sigma, \pi) = K(\sigma, \tau)$.
Thus the property of full-rankings from \cite{diaconis1977spearman} still holds for $\pi$ a top-list:
$K(\sigma,\pi) \leq F(\sigma, \pi) \leq 2 K(\sigma, \pi)$.
Letting $F(\sigma,p) := \sum_{\pi} p(\pi) F(\sigma, \pi)$ we have:
$K(\sigma,p) \leq F(\sigma,p) \leq 2 F(\sigma,p)$. 

\begin{figure}[h!]
\vspace{-.2cm}
  \begin{tikzpicture}[scale=.75]
    \draw (8.5, 0) node[anchor=west] {$\sigma \in \mathfrak T_8^8 = \mathfrak S_8$};
    \draw (8.5, 1) node[anchor=west] {$\pi \in \mathfrak T_8^4$};
    \draw (8.5, 2) node[anchor=west] {$\tau \in \mathfrak T_8^8 = \mathfrak S_8$};
    \foreach \y in {-.1,.1,.9,1.1,1.9,2.1} {\draw (0,\y) -- (8.2,\y);}
    \foreach \x in {0,1,...,8}{\draw (\x,1.9) -- (\x,2.1);}
    \foreach \x in {0,1,...,8}{\draw (\x,-.1) -- (\x,.1);}
    \foreach \x in {0,1,...,4}{\draw (\x,0.9) -- (\x,1.1);}

    \foreach \x in {.5,1.5,..., 7.5}{\fill (\x,0) circle (1.5pt);}
    \foreach \x in {.5,1.5,..., 7.5}{\fill (\x,1) circle (1.5pt);}
    \foreach \x in {.5,1.5,..., 7.5}{\fill (\x,2) circle (1.5pt);}
    \foreach \a/\b in {4.5/1.5,7.5/7.5,2.5/0.5,0.5/2.5,1.5/4.5, 
                       3.5/6.5,5.5/3.5,6.5/5.5}{
      \draw (\a,2) -- (\a,1) -- (\b,0);
    }
    \draw[line width=1.5pt,->,>=stealth] (0.5,1) -- (2.5,0);
    \draw[line width=1.5pt,->,>=stealth] (1.5,1) -- (4.5,0);
    \draw[line width=1.5pt,->,>=stealth] (3.5,1) -- (6.5,0);
  \end{tikzpicture}
  \vspace{-.6cm}
\caption{Representation of the generalized Spearman's footrule distance between $\pi$ and $\sigma$. Here $F(\sigma, \pi) = 16$, and the three candidates that contribute to the cost are materialized by three arrows.}
\vspace{-.2cm}
\label{figure:footrule}
\end{figure}
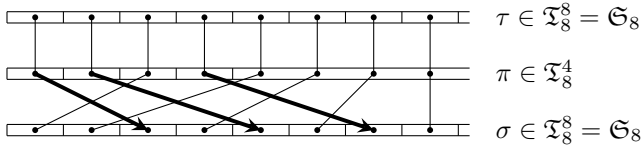

To generalize the footrule algorithm from \cite{dwork2001rank},
we need to express  $F(\sigma, \pi)$ as a sum over $i\in [n]$ of the cost
of putting candidate $i$ at rank $\sigma_i$. We first observe that the sum
of displacements in one direction is equal to the sum of displacements in the other,
%$\sum_{i=1}^n (\tau_i - \sigma_i) \cdot \mathbb 1_{\sigma_i < \tau_i}
%= \sum_{i=1}^n (\sigma_i - \tau_i) \cdot \mathbb 1_{\tau_i < \sigma_i}$.
 thus 
$F(\sigma, \tau) = 2 \sum_{i=1}^n  (\sigma_i - \tau_i) \cdot \mathbb 1_{\tau_i < \sigma_i}$.
Note that if $\pi_i < \infty$ then $\tau_i = \pi_i$;
and if $\pi_i = \infty$, then $\tau_i \geq \sigma_i$.
Thus $\mathbb 1_{\tau_i < \sigma_i} = \mathbb 1_{\pi_i < \sigma_i}$ and 
$F(\sigma, \pi) = 2 \sum_{i=1}^n (\sigma_i - \pi_i) \cdot \mathbb 1_{\pi_i < \sigma_i}$. Hence:
\begin{align*}
F(\sigma,p) &=
\sum_{\pi \in \mathfrak T_n} p(\pi) F(\sigma, \pi)\\
\nonumber &= 2 \sum_{i \in [n]} \sum_{\pi \in \mathfrak T_n}  p(\pi) \cdot (\sigma_i - \pi_i) \cdot \mathbb 1_{\pi_i < \sigma_i}\\
\nonumber &= 2 \sum_{i \in [n]} \sum_{k=1}^{\sigma_i} p(\pi_i = k) \cdot (\sigma_i - k) 
\end{align*}
Because of that, Algorithm~\algofootrule\ is able to optimize $F(\sigma,p)$ by solving a min-cost-perfect-matching problem.
For any full-ranking $\sigma^*$ we have $\cost{\sigma}{p} \leq F(\sigma,p) \leq F(\sigma^*,p) \leq 2 \cost{\sigma^*}{p}$.
Hence, Algorithm~\algofootrule\ is a 2-approximation for \textsc{Top-Agg}.
The time complexity is the time complexity of the Hungarian algorithm, which
computes a minimum-weight-perfect-matching.
\end{proof}

\section{Scores and average ranks}
\label{section:repeatchoice}
In this section we introduce the \emph{scores} and \emph{average ranks} of candidates.
Those two parameters are central to the problem of top-list aggregation.

When aggregating full-rankings, it is folklore that outputting a full-ranking
randomly sampled from the input gives an expected 2-approximation. 
%The proof
%relies on the fact that Kendall's tau distance on full-rankings satisfy the 
%triangle inequality. 
Ailon generalized this into design algorithm
\textsc{RepeatChoice}~\cite{ailon2010aggregation}, which is a 2-approximation
in the more general setting of bucket-order aggregation.
Algorithm \algorandomsort\ below is algorithm \textsc{RepeatChoice} specialized to \textsc{Top-Agg} and  reinterpreted using exponential
random variables.

\begin{algo}{\textbf{Algorithm} \textsc{RandomSort}}
  \label{algo:randomsort}
  \begin{algorithmic}
    \State {\bf Input: } instance $(n,p)$ of \textsc{Top-Agg}
    \State For each top-lists $\pi$ of the voting profile $p$:
    \State\quad Draw a real value $X_\pi$ from an exponential
    distribution of parameter $p(\pi)$.
    \State For each candidate $i$ in $[n]$:
    \State\quad Consider tuples $(X_\pi, \pi_i)$ with $\pi$ such that $i$ is a top candidate.
    \State\quad Choose $t_i$ to be the one with smallest value of $X_\pi$. 
    \State Build a full ranking $\sigma$, sorting the candidates using the lexicographical order over the $t_i$'s. 
    \State {\bf Output} $\sigma$.
  \end{algorithmic}
\end{algo}

For example, if we take the instance from Figure~\ref{figure:instance},
Algorithm \algorandomsort\ randomly orders the top-lists $\pi_1$,
$\pi_2$, $\pi_3$ and $\pi_4$,
by sorting top-lists by increasing order of their values $X_\pi$.
With probability $4/35 = 3/10 \cdot 4/7 \cdot 2/3$, the ordering is
$x_{\pi_3} < x_{\pi_4} < x_{\pi_2} < x_{\pi_1}$.
Observe that sorting candidates (by the values of their tuples) is equivalent
to processing the top-lists in order, appending candidates sequentially :
from $\pi_3$, we append candidates $4,1,5,2$;
then from $\pi_4$ we append candidates $6,3$;
then from $\pi_2$ we append no candidate;
then from $\pi_1$ we append candidate $7$;
then we append $8$ who is the only remaining candidate.
The resulting full-ranking $\sigma = [4,1,5,2,6,3,7,8]$ is at a distance
$K(\sigma, p) = 5.9$ from $p$.
This algorithm is a 2-approximation, but observe that candidate 1 is
never ranked first (even thought that would have been optimal).

\thmrandomsort %restatable, defined in intro.tex

\begin{proof}
  The time complexity is studied in the standard randomized real RAM model.
  Let $i$ and $j$ be two distinct candidates, each appearing at least once as
  a top candidate in the input voting profile.
  We compute the probability (over the values of the $X_\pi$)
  that $\sigma_i > \sigma_j$.
  Let $I = \min\{X_\pi:\pi_i < \pi_j\}$ and $J = \min\{X_\pi:\pi_j < \pi_i\}$
  be the minimum values of the exponential random variables over the sets
  of top-lists which respectively prefers $i$ to $j$ and $j$ to $i$.
  Observe that $\sigma_i > \sigma_j$ if and only if $I > J$.
  As the minimum of several exponential random variables is an exponential random
  variable with a parameter equal to the sum of parameters, $I$ and $J$ are two
  independent exponential random variables of parameters $p(\pi_i < \pi_j)$
  and $p(\pi_j < \pi_i)$. 
  Thus the probability that $\sigma_i > \sigma_j$ is
  $\mathbb P(\sigma_i > \sigma_j) =
  p(\pi_i > \pi_j) / p(\pi_i \neq \pi_j)$.
  We now compute the expected cost of the output $\sigma$.
  \begin{align*}
  \mathbb E(\cost{\sigma}{p}) &= 
  \sum_{(i,j) \in [n]^2} p(\pi_i < \pi_j) \cdot
  \mathbb E(\mathbb 1_{\sigma_i > \sigma_j})\\
  \nonumber &\leq 2 \sum_{\{i,j\}\subseteq [n]}\min \left\{\begin{array}{l}p(\pi_i < \pi_j)\\p(\pi_i > \pi_j)\end{array}\right.
  \end{align*}
  Let $\sigma^*$ denote the optimal solution. For all
  distinct  $i,j \in [n]$, 
  $\sigma^*$ must rank $i$ before $j$ or $j$ before $i$, 
  which costs at least the minimum between $p(\pi_i < \pi_j)$ and $p(\pi_i > \pi_j)$.
  Therefore $\mathbb E(\cost{\sigma}{p}) \leq 2\cost{\sigma^*}{p}$.
\end{proof}

Observe that in Algorithm \algorandomsort\ 
for any candidate $i$, the expected value of his tuple can be computed  easily.
Indeed,  the first coordinate of his tuple is the minimum of several
exponential random variables (all $X_\pi$ such that $i$ is a top candidate in $\pi$);
thus it is an exponential random variable whose parameter is $p(\pi_i < \infty)$.
As for the second coordinate,
we can easily compute the probability that an exponential random variable $X_\pi$
is smaller than all the exponential random variables of top-lists having $i$
as a top candidate: this probability is $p(\pi)/p(\pi_i < \infty)$, which
directly gives the expected value of the second coordinate.
\[\mathbb E[t_i] = \left(\frac{1}{p(\pi_i < \infty)},\; \sum_{r=1}^n \frac{p(\pi_i = r)}{p(\pi_i < \infty)} \cdot r \right )\]

From this observation we define the \emph{score} and \emph{average rank} of a candidate.
The score is known in the literature as the \emph{approval score} under a voting profile
that ignores the ordering between top candidates (in the setting where, instead of ranking top candidates, each voter gives a subset of approved candidates).

\begin{definition}[Score, Average rank]
Given a voting profile $p$, the \emph{score} of a candidate $i \in [n]$ is the
probability that she is a top candidate:
$\score{i}{p} := p(\pi_i < \infty)$.  Assuming that each candidate
appears at least once as a the top candidates in the input, 
The \emph{average rank} of a candidate is her expected rank,
conditioning on her being a top candidate:
$\rank{i}{p} := \sum_{r=1}^n \frac{p(\pi_i = r)}{\score{i}{p}} \cdot r$.
\end{definition}

\begin{figure}[h]
\[\small\begin{array}{|c|cccccccc|}\hline
i & 1 & 2 & 3 & 4 & 5 & 6 & 7 & 8 \\\hline
\score{i}{p} & 10/10 & 7/10 & 7/10 & 5/10 & 6/10 & 4/10 & 1/10 & 0/10 \\\hline
\rank{i}{p} & 21/10 & 24/7 & 19/7 & 9/5 & 19/6 & 4/4 & 4/1 & - \\\hline
\end{array}\]
\caption{Scores and average ranks of candidates in the instance
from Figure~\ref{figure:instance}.}
\end{figure}

\section{Generalized Borda's algorithm}
\label{section:borda}
\begin{algo}
{\textbf{Algorithm} \textsc{Borda$+$} (Generalization of \textsc{Borda})}
  \label{algo:borda}
  \begin{algorithmic}
    \State {\bf Input: } instance $(n,p)$ of \textsc{Top-Agg}
    \State For each candidate $i \in [n]$, compute
    $\rank{i}{p} \gets \sum_{r=1}^n \frac{p(\pi_i = r)}{p(\pi_i < \infty)} \cdot r$
    \State Sort candidates by increasing value of $\rank{i}{p}$.
    \State {\bf Output} the resulting full-ranking $\sigma$.
  \end{algorithmic}
\end{algo}

In this section, we draw inspiration from two noteworthy papers that study 
the approximation ratios of simple algorithms for full-ranking aggregation.
In \cite{coppersmith2010ordering}, Coppersmith, Fleischer and Rudra proved
that Borda's method is a 5-approximation.
In \cite{fagin2016algorithmic}, Fagin, Kumar, Mahdian, Sivakumar and Vee designed 
a general framework to prove constant factor approximation bounds.

In Borda's method for full-ranking aggregation, a candidate ranked in $r$-th
position by an input ranking gets $n-r$ points, and then candidates are sorted by total number of points. This is equivalent to sorting candidates by increasing average ranks.
Thus Algorithm~\algoborda\ can be seen as a generalization of Borda's method to \textsc{Top-Agg},
where the analysis uses insights from \cite{coppersmith2010ordering} to bound the
approximation ratio when the scores of candidates are within a constant factor of each other.

Let us give an example of execution of Algorithm~\algoborda,
using the instance from Figure~\ref{figure:instance}. (Candidate $8$ never appears as a top candidate so it is ranked last
in the output). 
Sorting  the instance from Figure~\ref{figure:instance} by average ranks produces the full-ranking $[6,4,1,3,5,2,7,8]$,
which is at distance $\cost{\sigma}{p} = 6.3$ from $p$.

\thmborda %restatable, defined in intro.tex

\begin{proof}
Let $(n,p)$ be an instance of \textsc{Top-Agg},
 $\sigma$ be the output of Algorithm~\algoborda,
and let $\sigma^*$ be the full-ranking minimizing $\cost{\sigma^*}{p}$.
We define $F(\sigma,p)$ as in the proof of Theorem~\ref{theorem:footrule}.
To simplify notations, we also define the positive part function
$x \mapsto x^+ = x \cdot \mathbb 1_{x > 0}$.
\begin{equation*}
F(\sigma,p) = 2\sum_{i=1}^n \sum_{k=1}^n  p(\pi_i = k) \cdot (\sigma_i - k)^+
\end{equation*}
We have the triangle inequality:
for all $x,y \in \mathbb R$, $(x+y)^+ \leq x^+ + y^+$.
Thus for all $i,r$ we have
$(\sigma_i - r)^+ \leq (\sigma_i^* - r)^+ + (\sigma_i - \rank{i}{p})^+ + (\rank{i}{p} - \sigma_i^*)^+$.
Recalling that $\score{i}{p} = \sum_{r=1}^n p(\pi_i = r)$, we obtain an upper bound on $F(\sigma,p)$:
\begin{align*}
\textstyle
F(\sigma,p) \leq F(\sigma^*,p)
&+ 2\sum_{i=1}^n \score{i}{p} \cdot (\sigma_i - \rank{i}{p})^+\\
&+ 2\sum_{i=1}^n \score{i}{p} \cdot (\rank{i}{p} - \sigma_i^*)^+
\end{align*}
One can prove 
(\textit{e.g.} Lemma 3.5 from \cite{coppersmith2010ordering})
that sorting by increasing average rank minimizes
$\sum_{i=1}^n (\sigma_i - \rank{i}{p})$.
Using the fact that the scores are all within a factor $\alpha$ of each others, we have:
\begin{align*}
&\sum_{i=1}^n \score{i}{p} \cdot (\sigma_i - \rank{i}{p})^+ \\
&\leq
\underbrace{
\left(\max_{i\in[n]}\score{i}{p}\right)
}_{\leq \alpha\min_{i\in[n]} \score{i}{p}}\cdot
\underbrace{
\sum_{i=1}^n (\sigma_i - \rank{i}{p})^+
}_{\leq \sum_{i=1}^n (\sigma^*_i - \rank{i}{p})^+}
\\&\leq \alpha\sum_{i=1}^n \score{i}{p} \cdot (\sigma^*_i - \rank{i}{p})^+
\end{align*}
Using this inequality to bound $F(\sigma, p)$, we obtain:
\[
F(\sigma,p) \leq F(\sigma^*,p)
+ 2\alpha\sum_{i=1}^n \score{i}{p} \cdot |\sigma_i^* - \rank{i}{p}|
\]
For all $i \in [n]$, we use the convexity of $x \mapsto |\sigma_i^* - x|$
and the definition $\rank{i}{p} = \sum_{k=1}^n \frac{p(\pi_i = k)}{\score{i}{p}} \cdot k$.
\begin{align*}
\sum_{i=1}^n \score{i}{p} \cdot |\sigma_i^* - \rank{i}{p}|
&\leq \sum_{i=1}^n\sum_{k=1}^n p(\pi_i = k) \cdot |\sigma_i^* - k|
\\&\leq F(\sigma^*,p)
\end{align*}
Combining the last two inequalities, we obtain $F(\sigma,p) \leq (1+2\alpha) F(\sigma^*,p)$.
Using the relation between $F$ and $K$, we conclude with
$\cost{\sigma}{p} \leq F(\sigma,p) \leq (1+2\alpha) F(\sigma^*,p) \leq (2+4\alpha) \cost{\sigma^*}{p}$.
\end{proof}

\paragraph{Tightness.} We notice that 
\algoborda is an $\Omega(\alpha)$ approximation in the worst case: let $n=2$ and let $p$ such that
$p([1;\dots]) = 0.999$ and $p([2,1;\dots]) = 0.001$; the optimal solution is
$[1,2]$ and costs $0.001$ whereas sorting by average ranks produces $[2,1]$ which
costs $0.999$. Thus, in general, \algoborda\ is not a $\mathcal O(1)$-approximation algorithm.

\medskip
Observe that sorting by
decreasing scores is not  a $\mathcal O(1)$-approximation algorithm either: let $n = 2$ and let $p$ such that $p([1,2;\dots]) = 0.999$ and $p([2;\dots]) = 0.001$;
the optimal solution is $[1,2]$ and costs $0.001$ whereas sorting by scores produces $[2,1]$ which
costs $0.999$. 

\medskip
However, in the next section we show that sorting first by
decreasing scores, then by increasing average ranks, yields an $\mathcal O(1)$-approximation algorithm

\section{Combining approval and Borda's methods}
\label{section:sorting}
\begin{algo}{\textbf{Algorithm} \textsc{Score-then-Borda$+$}}
  \label{algo:sorting}
  \begin{algorithmic}
      \algblockdefx{Indent}{EndIndent}[1]{{\bf #1}}{}\algtext*{EndIndent}
      \State {\bf Input: } an instance $(n,p)$ of \textsc{Top-Agg}
      \Indent{Step 1, partition candidates into intervals:}
        \State $u \gets$ uniformly random value on $[0,1)$.
        \ForAll{candidate $i \in [n]$}
          \State Compute $\score{i}{p} \gets p(\pi_i < \infty)$.
          \State Set $t \gets \lfloor u - \ln (\score{i}{p})\rfloor$ and put candidate $i$ in interval $E_t$.
        \EndFor
      \EndIndent
      \Indent{Step 2, solve the problem in each interval:}
        \ForAll{$t \in \mathbb N \cup \{\infty\}$ such that $E_t$ is non-empty}
          \State Order $E_t$ sorting candidates $i$ by average rank
          $\rank{i}{p} \gets \sum_{r=1}^n \frac{p(\pi_i = r)}{p(\pi_i < \infty)} \cdot r$.
        \EndFor
      \EndIndent
      \State Concatenate the ranking of $E_0$, ranking of $E_1$, $\dots$, and ranking of $E_\infty$.
      \State {\bf Output} resulting full-ranking.
  \end{algorithmic}
\end{algo}

In the previous section, we saw that when all scores are within a constant
factor of each other, then sorting by average rank yields a constant factor
approximation. In this section we argue that we can always do an approximate sort of the candidates
using rough scores, and then obtain a
constant factor approximation. This statement is made
more precise in Lemma~\ref{lemma:partition}, and used in Theorem~\ref{theorem:sorting}
to prove that Algorithm~\algosorting\ is a constant factor approximation.

Let us give an example of the execution of Algorithm~\algosorting,
using the instance from Figure~\ref{figure:instance}. 
In the first step, we sample a random value $u$ from [0,1), for example $u = 0.4$,
and use this value to define thresholds on the scores:
\begin{itemize}
\item a candidate $i$ such that
$0.55 \approx \exp(u-1) \leq \score{i}{p}$
will go in interval $E_0$;
\item a candidate $i$ such that
$0.20 \approx \exp(u-2) \leq \score{i}{p} <
\exp(u-1) \approx 0.55$
will go in interval $E_1$;
\item a candidate $i$ such that
$0.07 \approx \exp(u-3) \leq \score{i}{p} <
\exp(u-2) \approx 0.20$
will go in interval $E_2$;
\item and a candidate $i$ such that
$\score{i}{p} = 0$ will go in interval $E_\infty$.
\end{itemize}
At the end of the first step we have $E_0 = \{1,2,3,5\}$, $E_1 = \{4,6\}$,
$E_2 = \{7\}$ and $E_\infty = \{8\}$.
In the second step, we reorder candidates by increasing average ranks:
the ordering of $E_1$ is $[1,3,5,2]$; the ordering of $E_1$ is $[6,4]$;
the ordering of $E_2$ is $[7]$; the ordering of $E_\infty$ is $[8]$.
Finally, we concatenate the rankings of $E_0$, $E_1$, $E_2$ and $E_\infty$.
We obtain a full-ranking $\sigma = [1,3,5,2,6,4,7,8]$ which is at a distance
$\cost{\sigma}{p} = 5.8$ from $p$.

\thmsorting %restatable, defined in intro.tex

\begin{lemma}
\label{lemma:partition}
  Consider a constant $\eta > 0$ and an instance $(n,p)$ of \textsc{Top-Agg}.
  Sample a random variable $u$ uniformly at random from $[0,1)$.
  Define a partition function $f: s \mapsto  \lfloor u - \eta \ln (s)\rfloor$.
  A full-ranking $\sigma$ respects the partition if for any two candidates $i$ and $j$,
  having $f(\score{i}{p}) < f(\score{j}{p})$ implies that $\sigma_i < \sigma_j$.
  The expected cost of the best full-ranking that respects the partition
  is at most $(1+\eta)$ times the cost of the overall optimal full-ranking.
\end{lemma}

\begin{proof}
For all $t \in \mathbb N$ we define $E_t$ to be the set of candidates
that are sent in the $t$-th interval by the partition function.
Let $\sigma^*$ be an optimal solution and let $\sigma'$ be the full-ranking
which is closest to $\sigma^*$ and respects the partition.
More precisely, for all $t \in \mathbb N$, the full-ranking $\sigma^*$ induces
an ordering of the candidates from $E_t$;
we build $\sigma'$ as a concatenation of those rankings.
The cost of the best full-ranking that respects the partition is smaller than $\cost{\sigma'}{p}$.
From the definition of cost, we have:
\begin{align*}
&\cost{\sigma'}{p} - \cost{\sigma^*}{p} =\\
&\sum_{i \in [n]} \sum_{j \in [n]} \mathbb 1_{\sigma'_i > \sigma'_j} \cdot \mathbb 1_{\sigma^*_i < \sigma^*_j} \cdot \Big(\underbrace{
p(\pi_i < \pi_j) - p(\pi_j < \pi_i)}_{\text{smaller than } \score{i}{p}}
\Big)
\end{align*}
Let $i$ and $j$ be two candidates such that $\sigma^*_i < \sigma^*_j$.
Observe that having $\sigma'_i > \sigma'_j$ implies that $\score{i}{p} < \score{j}{p}$;
thus we assume the later.
We are going to compute the probability (over the randomness $u$) that $\sigma'_i > \sigma'_j$.
Candidates $i$ and $j$ are not in the same interval if and only if
\[\exists t \in \mathbb N,\quad
t + \eta\ln(\score{i}{p}) \leq u < t + \eta\ln(\score{j}{p})\]
This happens with probability at most $\eta\ln(\score{j}{p}/\score{i}{p})
\leq \eta(\score{j}{p}/\score{i}{p} - 1)$.
\\Hence:
\begin{align*}
&\mathbb E_u[\cost{\sigma'}{p}] - \cost{\sigma^*}{p} \\
&\leq \sum_{i \in [n]} \sum_{j \in [n]}
\mathbb 1_{\sigma^*_i < \sigma^*_j} \cdot \mathbb E_u[\mathbb 1_{\sigma'_i > \sigma'_j}\cdot \score{i}{p}]\\
&\leq \eta\sum_{i \in [n]} \sum_{j \in [n]}
\mathbb 1_{\sigma^*_i < \sigma^*_j}
\cdot\underbrace{(\score{j}{p} - \score{i}{p})^+}
_{\text{smaller than }p(\pi_j < \pi_i)}\\
&\leq \eta\cost{\sigma^*}{p}
\end{align*}
Observe that $\score{j}{p} - \score{i}{p}$ is a lower bound on the
weight of top-lists for which $j$ is a top candidate but $i$ is not.
We recognize a lower bound on the cost of $\sigma^*$, thus $\mathbb E_u[\cost{\sigma'}{p}] \leq (1+\eta) \cost{\sigma^*}{p}$.
\end{proof}

\begin{proof}\emph{(Theorem~\ref{theorem:sorting})}
  Let $(n,p)$ be an instance of $\textsc{Top-Agg}$, let $\sigma^*$ be an optimal solution
  and let $\sigma$ be the output of Algorithm \algosorting.
  The proof of this theorem is in two parts, corresponding to the two steps of
  the algorithm.
  
  Firstly, let $u$ be the random variable sampled during the first step,
  and let $\sigma'$ be the best full-ranking that respects the partition.
  From Lemma~\ref{lemma:partition} with $\eta = 1$, we have
  $\mathbb E_u[\cost{\sigma'}{p}] \leq 2 \cost{\sigma^*}{p}$.
  
  Secondly, we reuse the proof of Theorem~\ref{theorem:borda},
  with some additional details: every full-ranking that we consider needs to
  respect the partition (hence we replace every instance of $\sigma^*$ by $\sigma'$).
  On every interval, the ratio between the largest and smallest score is upper-bounded
  by $\alpha = e$; thus we have $\cost{\sigma}{p} \leq (4e+2) \cost{\sigma'}{p}$.
  
  Combining both parts, we obtain that Algorithm \algosorting\ is
  a randomized $(8e+4)$ approximation. Note that we did not try to optimize
  the approximation ratio.
\end{proof}

\section{PTAS for top-list aggregation}
\label{section:ptas}
In the case of full-ranking aggregation, \cite{kenyon2007rank, mathieu2009rank}
show that there is a PTAS. The approximation scheme with the best running time
is algorithm \textsc{FASTer-Scheme} from \cite{mathieu2009rank}.
Rephrasing it to the setting of top-list aggregation,
it requires that all candidates are compared a similar number of times
(Theorem~\ref{theorem:fastlight} makes this statement more precise).
We notice that this condition is equivalent with having all the scores within
a constant factor of each other; therefore we can use the techniques from
the previous section.

\smallskip
At a high level, Algorithm~\algoptas\ starts by fixing thresholds on the
scores of candidates (exactly as Algorithm~\algosorting\ does), to partition
candidates into intervals. Then it uses \textsc{FASTer-Scheme} as a black-box,
to find a nearly-optimal solution on each interval.
We show in Theorem~\ref{theorem:ptas} that Algorithm~\algoptas\ is a PTAS
for top-list aggregation.

\begin{algo}{\textbf{Algorithm} \textsc{Score-then-PTAS} with error parameter $\varepsilon > 0$.}
  \label{algo:ptas}
  \begin{algorithmic}
      \algblockdefx{Indent}{EndIndent}[1]{{\bf #1}}{}\algtext*{EndIndent}
      \State {\bf Input: } an instance $(n,p)$ of \textsc{Top-Agg}
      \Indent{Step 1, partition candidates into intervals:}
        \State $u \gets$ uniformly random value on $[0,1)$.
        \ForAll{candidate $i \in [n]$}
          \State Compute $\score{i}{p} \gets p(\pi_i < \infty)$.
          \State Set $t \gets \lfloor u - (\varepsilon/3) \ln (\score{i}{p})\rfloor$ and put candidate $i$ in interval $E_t$.
        \EndFor
      \EndIndent
      \Indent{Step 2, solve the problem in each interval:}
        \ForAll{$t \in \mathbb N \cup \{\infty\}$ such that $E_t$ is non-empty}
          \State $p_t \gets \text{restriction of input top-lists to }E_t$
          \State Order $E_t$ using \textsc{FASTer-Scheme} \cite{mathieu2009rank} on instance $p_t$ with error parameter $\varepsilon/3$.
        \EndFor
      \EndIndent
      \State Concatenate the ranking of $E_0$, ranking of $E_1$, $\dots$, and ranking of $E_\infty$.
      \State {\bf Output} resulting full-ranking.
  \end{algorithmic}
\end{algo}

Let us give an example of execution of Algorithm~\algoptas\  
with $\varepsilon = 3$,
using the instance from Figure~\ref{figure:instance}.
In the first step, we sample a random value $u$ from [0,1), for example $u = 0.4$.
Observe that we chose $\varepsilon$ and $u$ such that  the partition in interval
is the same as in Algorithm \algosorting:
$E_0 = \{1,2,3,5\}$, $E_1 = \{4,6\}$, $E_2 = \{7\}$ and $E_\infty = \{8\}$.
In the second step, we find an approximate solution of the optimal ordering of every non-empty interval.
For $E_0$, we build the restriction $p_0$ of the voting profile $p$ on 
candidates from $E_0$: here $\pi_1 = [3,5,1;\dots]$, $\pi_2 = [3,1,5;\dots]$,
$\pi_3 = [1,5,2;\dots]$ and $\pi_4 = [1,2,3;\dots]$. Then we use the algorithm
\textsc{FASTer-Scheme} to find a $(1+\varepsilon/3)$
approximation of the optimal solution for $p_0$, and get (for example)
the ranking $[1,2,3,5]$. We do the same for $E_1$, $E_2$ and $E_\infty$
and get (for example) the rankings $[4;6]$, $[7]$ and $[8]$.
Finally, we concatenate the rankings of $E_0$, $E_1$, $E_2$ and $E_\infty$.
We obtain a full-ranking $\sigma = [1,2,3,5,4,6,7,8]$ which is at a distance
$\cost{\sigma}{p} = 5.5$ from $p$.

\theoremptas

\begin{theorem}[from \cite{mathieu2009rank} Theorem~1.2]
\label{theorem:fastlight}
Let $b \in (0,1]$ be a parameter. There exists a randomized polynomial time approximation scheme
(called \textsc{FASTer-Scheme} in \cite{mathieu2009rank}) for the special case of \textsc{Top-Agg}  such that the input  $(n,p)$ satisfies 
\begin{equation*}
\min_{\substack{i,j \in[n] \\ i \neq j}}\; p(\pi_i \neq \pi_j) \geq b \cdot \max_{\substack{i,j \in[n] \\ i \neq j}}\; p(\pi_i \neq \pi_j)
\end{equation*}
The running time\footnote{Recall that
  $f(x) = \tilde{\mathcal O}(g(x))$ if there is a constant $\ell$ such that
  $f(x) = \mathcal O(g(x)\log^\ell (g(x)))$\label{footnote}} is
  $\mathcal O((\log(\frac{1}{b}) + \frac{1}{\varepsilon}) \cdot n^3 \log n)
  + n2^{\tilde{\mathcal O}(1 / (\varepsilon b)^6)}$.
  The algorithm can be derandomized, but
  $n^{\tilde{\mathcal O}(1/(\varepsilon b)^{12})}$
  is added to the running time.
\end{theorem}

\begin{proof}\emph{(Theorem~\ref{theorem:ptas})}
Let $0 < \varepsilon \leq 3$, let $(n,p)$ be an instance of $\textsc{Top-Agg}$,
let $\sigma^*$ be an optimal solution, and let $\sigma$ be the output of
Algorithm~\algoptas\ on $(n,p)$ with error parameter $\varepsilon$.
The proof of this theorem is in two parts, corresponding to the two steps of
Algorithm~\algoptas.

Firstly, we prove that there exists a ranking $\sigma'$ whose expected cost (over the randomness of $u$) is nearly-optimal
and that \emph{respects} the partition $(E_t)$, in the sense that all candidates of $E_t$ precede all candidates of $E_{t+1}$.
To prove that, let us use Lemma~\ref{lemma:partition} with
$\eta = \varepsilon/3$, and directly obtain that $\mathbb E_u[\cost{\sigma'}{p}] \leq (1+\varepsilon/3) \cost{\sigma^*}{p}$.

Secondly, we prove that a nearly-optimal solution for each interval
$E_t$ can be computed by algorithm \textsc{FASTer-Scheme} from \cite{mathieu2009rank}.
Let $t$ such that $E_t$ is non-empty,
and let $i$ and $j$ be two candidates from $E_t$. 
To prove that we are allowed to use \textsc{FASTer-Scheme} on instance $p_t$,
observe that $p_t(\pi_i \neq \pi_j)=p(\pi_i \neq \pi_j)$.
The critical remark here is that for all distinct $i,j \in [n]$ we have
\[\max(\score{i}{p}, \score{j}{p}) \leq p(\pi_i \neq \pi_j) \leq \score{i}{p} + \score{j}{p}.\]
Therefore, by definition of $E_t$, we have
\[\max_{\substack{i,j \in E_t\\i \neq j}} p(\pi_i \neq \pi_j)\leq 2\max_{i\in E_t}\score{i}{p}
< \exp({\textstyle\frac{u-t}{\varepsilon/3}})
\]
\[
\min_{\substack{i,j \in E_t\\i \neq j}} p(\pi_i \neq \pi_j)\geq\min_{i\in E_t}\score{i}{p}
\geq \exp({\textstyle\frac{u-t-1}{\varepsilon/3}}).\]
Thus, the equation in Theorem~\ref{theorem:fastlight} holds with $1/b := \exp(3/\varepsilon)$.
Algorithm~\textsc{FASTer-Scheme} produces a $(1+\varepsilon/3)$ approximation
of the optimal solution of $p_t$, in time
$\textstyle\mathcal O(|E_t|^3 \log |E_t|(\underbrace{\textstyle\log(\frac{1}{b})}_{\mathcal O(1/\varepsilon)} + \frac{1}{\varepsilon}))
+ |E_t|\hspace{-.3cm}\underbrace{2^{\tilde{\mathcal O}(1 / (\varepsilon b)^6)}}_{\exp(\exp(\mathcal O(1/\varepsilon)))}$.

To finish the proof, observe that Algorithm~\algoptas\ outputs a full-ranking
whose expected cost (over the randomness of \textsc{FASTer-Scheme}) is at most
$(1+\varepsilon/3)\cost{\sigma'}{p}$.
Thus the expected cost (over the randomness of $u$) is at most $(1+\varepsilon/3)^2 \cost{\sigma^*}{p}$.
For all $0 < \varepsilon \leq 3$, we have $(1+\varepsilon/3)^2 \leq (1+\varepsilon)$.

 To derandomize Algorithm~\algoptas, use in step 2 the derandomized version of \textsc{FASTer-Scheme};
 and in step 1 try all possible values of $u$ for which at least one of thresholds $(\exp(\frac{u-t}{\varepsilon/3}))_{t\in\mathbb N}$
 on the scores is equal to the score of one of the candidates,
 and output the best of all rankings thus computed.
\end{proof}

\section{EPTAS for top-$k$-list aggregation.}
\label{section:eptas}
In this section we study $\textsc{Top-$k$-Agg}$, the special case of \textsc{Top-Agg}
when all top-lists in the input voting profile have exactly $k$ top candidates.
The main result of this section is Theorem~\ref{theorem:eptas}, which
proves that Algorithm~\algoeptas\ is an EPTAS.

\begin{algo}{\textbf{Algorithm} \textsc{Score-then-Adjust}}
  \label{algo:eptas}
  \begin{algorithmic}
  \State {\bf Input: } Instance $(n,p)$ of \textsc{Top-$k$-Agg}.
      \State For each candidate $i \in [n]$, compute $\score{i}{p} \gets p(\pi_i < \infty)$.
      \State Define $\sigma'$, the full-ranking obtained by sorting candidates by non-increasing scores.
      \State By permuting the first
      $\lceil (1+\frac{1}{\varepsilon})(k-1)\rceil$ candidates of $\sigma'$,
      choose $\sigma$ which minimizes $K(\sigma, p)$.
      \State {\bf Output} $\sigma$.
  \end{algorithmic}
\end{algo}

Let us give an example of execution of Algorithm~\algoeptas\ 
with $n = 8$, $k = 4$ and $\varepsilon = 3$,
using the instance from Figure~\ref{figure:instance}.
First, we sort candidates by non-increasing scores,
which (for example) gives us the full-ranking $\sigma' = [1,3,2,5,4,6,7,8]$.
Then consider the first $m = \lceil (1+\frac{1}{\varepsilon})(k-1)\rceil = 4$
candidates of $\sigma'$, who are $1,3,2$ and $5$.
To compute the cost of a reordering of those candidates, we just need to consider
a restricted instance on those $m$ candidates\footnote{In this new instance,
top-lists of the voting profile does not necessarily have all the same size.}:
here $\pi_1 = [3,5,1;\dots]$, $\pi_2 = [3,1,5;\dots]$,
$\pi_3 = [1,5,2;\dots]$ and $\pi_4 = [1,2,3;\dots]$.
To find the optimal solution\footnote{One could use a polynomial time
approximation algorithm here, but the resulting algorithm will not be a PTAS},
one can enumerate the $m!$ possible full-rankings, or use a dynamic programming
approach and compute the optimal ordering of each of the $2^m$ subsets of $\{1,3,2,5\}$.
With both methods we find that the optimal reordering is $[1,2,3,5]$.
Hence we have $\sigma = [1,2,3,5,4,6,7,8]$, which is at a distance
$\cost{\sigma}{p} = 5.5$ from $p$.

\theoremeptas

We begin the analysis with a simple observation. Consider a full-ranking $\sigma$ and let $i$ such that $\sigma_i = n$. If an input top-$k$-list $\pi $ ranks $i$ among its top $k$ elements, then at least $n-k$ pairs $\{ i,j\}$ are ranked in reverse order in $\pi$ and in $\sigma$, so $K(\sigma, \pi) \geq n - k$. Thus 
$\cost{\sigma}{p}  \geq (n-k) \cdot p(\pi_i<\infty)=(n-k)\cdot Score(i)$. 
This observation can be generalized, leading to the statement of Lemma \ref{lemma:lb}.

\begin{lemma}\label{lemma:lb}
Let $(n,p)$ be an instance of \textsc{Top-$k$-Agg}.
For any full-ranking $\sigma^*$, we have a lower bound on the objective function:
$\cost{\sigma^*}{p}  \geq \sum_{i : \sigma^*_i > k} (\sigma^*_i - k) \cdot \score{i}{p} .$
\end{lemma}
\begin{proof}
% It remains to prove a lower bound on the optimal cost. 
We write %the cost of $\sigma^*$ as a sum over candidates $i \in [n]$.
\[\cost{\sigma^*}{p} 
= 
\sum_{i \in [n]}  \sum_{\pi \in \mathfrak T_n^k} 
\left|\left\{j \in [n] \;|\; \substack{\pi_i < \pi_j\\\sigma^*_j < \sigma^*_i}\right\}\right| \cdot p(\pi)
\]
Let $i \in [n]$ be a candidate and $\pi \in \mathfrak T_n^k$ be a top-list in which $i$ is a top candidate:
\[
\left|\left\{j \;|\; \substack{\pi_i < \pi_j\\\sigma^*_j < \sigma^*_i}\right\}\right|=
\Big|\Big\{j \;|\; \substack{\sigma^*_j < \sigma^*_i}\Big\}\Big| -
\left|\left\{j \;|\; \substack{\pi_j \leq \pi_i\\\sigma^*_j < \sigma^*_i}\right\}\right|
\geq \sigma^*_i-k.
\]
Thus, summing only over $i \in [n]$ such that $\sigma^*_i > k$ and over $\pi \in \mathfrak T_n^k$ for which $i$ is a top candidate.
\begin{align*}
\cost{\sigma^*}{p}  &\geq
\sum_{i:\sigma_i^* \geq k}
\sum_{\pi}
(\sigma^*_i-k) \cdot p(\pi) \cdot \mathbb 1_{\pi_i < \infty} \\&\geq
\sum_{i : \sigma^*_i > k} (\sigma^*_i - k) \cdot \score{i}{p}
\end{align*}
This conclude the proof.
\end{proof}

\begin{proof}\emph{(Theorem \ref{theorem:eptas})}
Let $\sigma$ a full-ranking and let $\pi$ a top-list.
The distance between $\sigma$ and $\pi$ can be split in two:
$K(\sigma, \pi) = K_{top}(\sigma, \pi) + K_{score}(\sigma, \pi)$.
\[
K_{top}(\sigma, \pi) :=
\sum_{i \in [n]} \sum_{j \in [n]} \mathbb 1_{\sigma_i > \sigma_j} \cdot \mathbb 1_{\pi_i < \pi_j < \infty}
\]
\[ K_{score}(\sigma, \pi) :=
\sum_{i \in [n]} \sum_{j \in [n]} \mathbb 1_{\sigma_i > \sigma_j} \cdot \mathbb 1_{\pi_i < \pi_j = \infty}
\]
The value $K_{top}(\sigma, \pi)$ is the number of inversions of top candidates
of $\pi$, between $\sigma$ and $\pi$; whereas $K_{score}(\sigma, \pi)$ can be
seen as the distance between $\sigma$ and a bucket-order with two buckets:
a partial order where all candidates of the first bucket (top candidates of $\pi$)
are ranked before candidates of the second bucket.

We now define $K_{top}(\sigma, p)$ and $K_{score}(\sigma, p)$ as weighted
averages of  $K_{top}(\sigma, \pi)$ and $K_{score}(\sigma, \pi)$,
over all top-lists $\pi$ of the voting profile $p$.
We have $K(\sigma, p) = K_{top}(\sigma, p) + K_{score}(\sigma, p)$.
One can observe that sorting candidates by decreasing score actually minimizes
$K_{score}(\sigma, p)$; this result is proved in Theorem~3 of
\cite{ailon2010aggregation}.
More precisely, whenever there are two candidates $i$ and $j$ such that
$\sigma_i = \sigma_j + 1$ and $\score{i}{p} > \score{j}{p}$,
the pair $(i,j)$ costs $p(\pi_i < \pi_j = \infty)$
in $K_{score}(\sigma, p)$. However, by definition of score:
\begin{align*}
& 0 < \score{i}{p} - \score{j}{p} \\&=
p(\pi_j < \pi_i < \infty) + p(\pi_i < \pi_j < \infty) + p(\pi_i < \pi_j = \infty)\\
&-p(\pi_i < \pi_j < \infty) - p(\pi_j < \pi_i < \infty) - p(\pi_j < \pi_i = \infty)\\
&= p(\pi_i < \pi_j = \infty) - p(\pi_j < \pi_i = \infty)
\end{align*}
Therefore swapping candidates $i$ and $j$ strictly decreases $K_{score}(\sigma, p)$.

Let $m:= \lceil (1+\frac{1}{\varepsilon})(k-1)\rceil$. Let $\sigma''$ denote the full-ranking obtained from $\sigma'$ by reordering the first $m$ candidates according to their relative order in the (unknown) optimal order~$\sigma^*$.
The algorithm outputs a full-ranking $\sigma$ such that
$\cost{\sigma}{p}\leq \cost{\sigma''}{p}$.

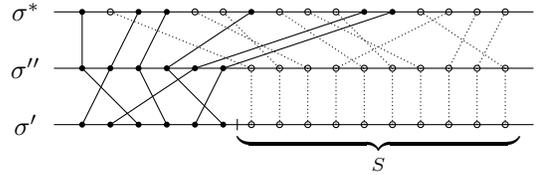
\begin{figure}[h!]
\begin{center}
\begin{tikzpicture}[scale=.75]
\draw (0,0) -- (8.5,0);
\draw (0,1) -- (8.5,1);
\draw (0,2) -- (8.5,2);
\draw (3.25,-.1) -- (3.25,.1);
\draw (0.5,2)--(0.5,1); \draw (1.5,0)--(0.5,1);
\draw (1.5,2)--(1.0,1); \draw (0.5,0)--(1.0,1);
\draw (2.0,2)--(1.5,1); \draw (2.0,0)--(1.5,1);
\draw (3.5,2)--(2.0,1); \draw (3.0,0)--(2.0,1);
\draw (5.5,2)--(2.5,1); \draw (1.0,0)--(2.5,1);
\draw (6.0,2)--(3.0,1); \draw (2.5,0)--(3.0,1);
\draw[densely dotted] (1.0,2)--(3.5,1); \draw[densely dotted] (3.5,0)--(3.5,1);
\draw[densely dotted] (3.0,2)--(4.0,1); \draw[densely dotted] (4.0,0)--(4.0,1);
\draw[densely dotted] (2.5,2)--(4.5,1); \draw[densely dotted] (4.5,0)--(4.5,1);
\draw[densely dotted] (7.0,2)--(5.0,1); \draw[densely dotted] (5.0,0)--(5.0,1);
\draw[densely dotted] (4.0,2)--(5.5,1); \draw[densely dotted] (5.5,0)--(5.5,1);
\draw[densely dotted] (5.0,2)--(6.0,1); \draw[densely dotted] (6.0,0)--(6.0,1);
\draw[densely dotted] (4.5,2)--(6.5,1); \draw[densely dotted] (6.5,0)--(6.5,1);
\draw[densely dotted] (7.5,2)--(7.0,1); \draw[densely dotted] (7.0,0)--(7.0,1);
\draw[densely dotted] (8.0,2)--(7.5,1); \draw[densely dotted] (7.5,0)--(7.5,1);
\draw[densely dotted] (6.5,2)--(8.0,1); \draw[densely dotted] (8.0,0)--(8.0,1);
%\draw [->,>=stealth] plot [smooth,tension=1] coordinates {(9,.1)(9.3,.5)(9,.9)};
%\draw [->,>=stealth] plot [smooth,tension=1] coordinates {(9,1.9)(9.3,1.5)(9,1.1)};%\node[anchor=west] at (9.5,.5) {Exhaustive search};
%\node[anchor=west] at (9.5,1.5) {Sequence of swaps};
\node at (-.5,2) {$\sigma^*$};
\node at (-.5,1) {$\sigma''$};
\node at (-.5,0) {$\sigma'$};
\node[anchor=north] at (5.75,0) {$\underbrace{\hspace{3.75cm}}_{S}$};
\foreach \x in {.5,1,...,3} \fill (\x,0) circle (1.5pt);
\foreach \x in {.5,1,...,3} \fill (\x,1) circle (1.5pt);
\foreach \x in {3.5,4,...,8} \draw (\x,0) circle (1.5pt);
\foreach \x in {3.5,4,...,8} \draw (\x,1) circle (1.5pt);
\foreach \x in {0.5,1.5,2.0,3.5,5.5,6.0} \fill (\x,2) circle (1.5pt);
\foreach \x in {1.0,3.0,2.5,7.0,4.0,5.0,4.5,7.5,8.0,6.5}
  \draw (\x,2) circle (1.5pt);
\end{tikzpicture}
\vspace{-.6cm}
\end{center}
\caption{Graphical representation of full-rankings $\sigma^*$, $\sigma''$ and $\sigma'$. Elements from $S$ are represented with light circles.}
\end{figure}

Letting $S$ denote the set of candidates of rank greater than $m$ in $\sigma'$, observe that $\sigma''$ can also be defined from $\sigma^*$ by doing a partial bubble sort, repeatedly swapping adjacent elements whenever their scores are out of order and at least one of the two is in $S$; thus $K_{score}(\sigma'', p) \leq K_{score}(\sigma^*, p)$. Moreover, as $\sigma''$ and $\sigma^*$ can disagree on the relative order of two candidates only if at least one of the two is in $S$:
%we can bound $K_{top}(\sigma'', p) - K_{top}(\sigma^*, p)$.
\begin{align*}
&K_{top}(\sigma'',p) - K_{top}(\sigma^*,p) \\&=
 \sum_{i,j \in [n]} \mathbb 1_{\left|\substack{\sigma''_i > \sigma''_j\\\sigma^*_i < \sigma^*_j}\right.} \cdot \Big(\underbrace{
p(\pi_i < \pi_j < \infty) - p(\pi_j < \pi_i < \infty)}_{\text{smaller than } p(\pi_i  < \infty \text{ and }\pi_j < \infty)}
\Big)\\
&\leq \sum_{s\in S}\sum_{t: t\neq s} p(\pi_s < \infty \hbox{ and }\pi_t < \infty).
\end{align*}
  Since the input consists of top-$k$-lists\footnote{Actually the Theorem also holds when the input top-lists have ties among the top $k$ candidates; the only thing that matters is that at least $n-k$ candidates are such that $\pi(i)=\infty$.}, for all $s \in S$ we have
\begin{align*}
&\sum_{t: t\neq s} p(\pi_s < \infty \hbox{ and }\pi_t < \infty)\\
&= \sum_{t:t\neq s} p(\pi_s < \infty)\cdot p(\pi_t < \infty\;|\;\pi_s < \infty)\\
&\leq (k-1)\cdot p(\pi_s < \infty)=(k-1)\cdot\score{s}{p}.
\end{align*}
Thus $K_{top}(\sigma'',p) - K_{top}(\sigma^*,p) \leq (k-1)\sum_{s\in S} \score{s}{p}$.
Since $S$ is also the set of $n-m$ elements with the smallest scores, we have
$\sum_{s\in S}\score{s}{p}\leq  \sum_{i: \sigma^*_i>m}\score{i}{p}$.
In summary, we proved that:
\[\cost{\sigma''}{p} - \cost{\sigma^*}{p} \leq (k-1) \sum_{i: \sigma^*_i>m}\score{i}{p}\]
Applying Lemma~\ref{lemma:lb}, with the fact that we always have $m \geq k$, 
\[\sum_{i: \sigma^*_i>m}\score{i}{p}
\leq \sum_{i: \sigma^*_i>m}\frac{(\sigma^*_i - k) \cdot \score{i}{p}}{m+1-k} \leq \frac{\cost{\sigma^*}{p}}{m+1-k}\]
Recalling  that $m = \lceil (1+\frac{1}{\varepsilon})(k-1)\rceil \geq k-1+ (k-1)/\varepsilon$ we finally obtain
\[
\cost{\sigma''}{p} \leq
\underbrace{\left(1+\frac{k-1}{m+1-k}\right)}_{\leq 1+\varepsilon} \cdot \cost{\sigma^*}{p}.
\]
 To achieve the claimed running time, we see that computing $\sigma'$ takes time $O(n\log n)$. To find the optimal reordering of the first $m$ candidates, we first precompute the values of $p(\pi_i < \pi_j)$ for all $i,j $ such that $\sigma'_i \leq m$ and $\sigma'_j \leq m$. Then, we use dynamic programming: for each subset of the first $m$ candidates of $\sigma'$ we try all possibilities for the candidate that will be ranked first, and store this candidate together with the cost of the associated solution. Time complexity is $O(m \cdot 2^m)$. Space complexity is exponential; if someone needs to be memory efficient the exhaustive search approach might be preferable.
\end{proof}
We remark that the running time of Algorithm \algoeptas\ is quasi-linear
as long as $k = o(\log n)$, and polynomial as long as $k=\mathcal O(\log n )$.

%\clearpage
%\appendix

%\section{Past and present choices of definitions}
%\label{section:metrics}
%\input{metric}

%\newpage
%\section{Proofs of Section~\ref{section:approxs}}
%\label{section:approx_proofs}
%\input{approxs_proofs}

%\newpage
%\section{Proofs of Section~\ref{section:ptas}}
%\label{section:ptas_proofs}
%\input{ptas_proofs}

%\newpage
%\section{Constant factor approximation for \textsc{Top-Agg}}
%\label{section:previous}
%\input{fas}
%\input{fast}

\newpage
\bibliographystyle{alpha}
\bibliography{biblio}

\end{document}